\documentclass{amsart}

\usepackage{mathtools}
\usepackage[square,numbers]{natbib}
\usepackage{float}
\restylefloat{table}
\usepackage[utf8]{inputenc}
\usepackage[toc,page]{appendix}
\usepackage{amsmath,amsthm,amssymb,stmaryrd}
\usepackage{graphicx}
\usepackage[colorinlistoftodos,prependcaption]{todonotes}
\usepackage{xargs}
\usepackage{multicol}
\newcommandx{\unsure}[2][1=]{\todo[linecolor=blue,backgroundcolor=blue!25,bordercolor=blue,#1]{#2}}
\usepackage[section]{placeins}
\usepackage{float}
\usepackage{comment}
\usepackage{esvect}
\usepackage{xcolor}
\usepackage{graphicx}
\usepackage{dcolumn}
\usepackage{bm}
\usepackage{amsaddr}

\newmuskip\pFqmuskip
\newcommand*\pFq[6][8]{%
	\begingroup 
	\pFqmuskip=#1mu\relax
	\mathcode`\.=\string"8000
	\begingroup\lccode`\~=`\,
	\lowercase{\endgroup\let~}\pFqcomma
	{}_{\,#2}F_{\,#3}{\left[\genfrac..{0pt}{}{\,#4}{\,#5};#6\right]}%
	\endgroup
}
\newcommand{\pFqcomma}{\mskip\pFqmuskip}

\newtheorem{theorem}{Theorem}[section]

\newtheorem{lemma}{Lemma}[section]

\newtheorem{definition}{Definition}[section]
\newcommand{\bbint}[2]{\ensuremath{\;\backslash\!\!\!\!\backslash\!\!\!\!\!\int_{#1}^{#2}}}

\begin{document}
\title[Generalized Stieltjes transform]{Exact Evaluation and extrapolation of the divergent expansion for the  Heisenberg-Euler Lagrangian I: Alternating Case}

\author{Christian D. Tica \& Eric A. Galapon}
\address{Theoretical Physics Group, National Institute of Physics, University of the Philippines, Diliman Quezon City, 1101 Philippines}
\email{eagalapon@up.edu.ph}
\date{\today}

\maketitle
\begin{abstract}
       We applied the method of finite-part integration [Galapon E.A Proc.R.Soc A 473, 20160567(2017)] to evaluate in closed-form the exact one-loop integral representations of the Heisenberg-Euler Lagrangian from QED for a constant magnetic field  and magnetic-like self-dual background. We also devise a prescription based on the finite-part integration of a generalized Stieltjes integral to sum and extrapolate to the strong-field regime the alternating divergent weak-field expansions of the Heisenberg-Euler Lagrangians.

\end{abstract}

\section{Introduction}

Divergent power series solutions of asymptotic nature arise in various context of perturbation theory (PT) when their exact or convergent alternatives are difficult if not impossible to obtain \cite{boyd, dingle, wong2, olver1997asymptotics, le2012large}. An optimally truncated asymptotic power series delivers rapidly converging and experimentally relevant approximations to the solution while standard summation techniques such as Pade approximants and Borel summation are employed to extend their numerical utility beyond the relevant asymptotic regime \cite{hardy, bender1999advanced}. Summation procedures also provide the framework within which further analysis can be made that reveals important computational and non-perturbative aspect of the necessarily finite solution \cite{wong1,mera2018fast,borell,sara}. A large class of divergent power series solutions in various applications involve terms that alternate in sign and are characterized as Stieltjes series,
\begin{equation}\label{mopy}
     F(\beta) \sim  \sum_{k=0}^{\infty}a_k(-\beta)^{k}, \qquad \beta\to 0^{+},
\end{equation}
where $\beta$ is the real and positive perturbation parameter and $a_k$ are the positive-power moments $\mu_k$ of some positive function $\rho(x)$, $a_k = \mu_{k} = \int_{0}^{\infty} x^k \rho(x)\mathrm{d}x$.
The summation of a finite physical quantity, $F(\beta)$, that is represented by a divergent Stieltjes series can be carried out by various means when  $\beta$ is small \cite{bender1999advanced, mera2018fast,weniger1993summation,kazakov2002summation, gilewicz2002continued, alvarez2017new}. 
Most notably, Pad\'{e} approximants constructed from a Stieltjes series possess well-understood convergence properties \cite{bender1999advanced, baker1996pade}. A summation procedure \cite{jen} based on non-linear (Weniger) sequence transformation also delivers rapidly converging results.  
 
Standard summation procedures by themselves, are largely ineffective at the  formidable task \cite{costin2019resurgent, costin2, bender1994determination, hydro, hydro2, suslov2001summing, le1990hydrogen, weniger1996construction} of extrapolating the divergent PT expansion \eqref{mopy} to the opposite regime $\beta\to+\infty$ along the real line, given the leading-order behavior of $F(\beta)$ as $\beta\to\infty$. The main hurdle is to incorporate this given specific behavior in the construction of a suitable extrapolant. 

In \cite{Tica_royal}, we constructed extrapolants from the divergent Stieltjes-type weak-coupling, $\beta\to 0$, PT expansion for the ground-state energy $E(\beta)$ of various singular eigenvalue problems in quantum mechanics.  These extrapolants are convergent expansions that naturally incorporate the non-integer algebraic power leading-order behaviour of the ground-state energy $E(\beta)\sim \beta^{\nu}, |\nu|<1$, as $\beta\to\infty$ allowing us to compute $E(\beta)$ for extremely large real coupling values. The treatment however needs to be modified when the leading-order non-perturbative behavior is not an algebraic power of the perturbation parameter $\beta$. A well-known example is the hydrogen atom in a constant magnetic field of magnitude $B$ \cite{ hydro, hydro2,le1990hydrogen} where the logarithmic behavior of the binding energy $\mathcal{E}(B)\sim\frac{1}{2} (\ln B)^2$ in the non-perturbative regime $B\to\infty$ becomes essential in large magnetic fields beyond the neutron stars range. This peculiar strong-field behavior renders the summation of the divergent PT expansion for the energy eigenvalues difficult to perform \cite{cizik}. 

In this paper, we consider the exact one-loop integral formulation of the Heisenberg-Euler Lagrangian for the case of a constant magnetic field background and a magnetic-like self-dual background in both scalar and spinor QED. 
For all these cases, we will derive a closed-form for the Lagrangian using a method that relies on the novel representations of Hadamard's finite part as complex contour integrals \cite{galapon2}. We give a concise discussion of this representation along with two other equivalent representations in section \ref{sorat} and carry out the derivation of the closed-form for the Heisenberg-Euler Lagrangian in section \ref{meth}.  

In section \ref{bigaj}, we take 
on the summation and extrapolation to the strong-field regime of the divergent weak-field PT expansions of the Heisenberg-Euler Lagrangian. The expansion coefficients, $a_k$, alternate in sign and possess a leading $a_k\sim(2k)!$ growth as $k\to\infty$. The Heisenberg-Euler Lagrangians also exhibit a known logarithmic leading-order behavior in the strong-field regime. The initial step in our procedure is to map the expansion coefficients, $a_k$, to the positive-power moments of some positive function $\rho(x)$, $a_k=\mu_{2k}=\int_{0}^{\infty}x^{2k} \rho(x) \mathrm{d}x$ so that we sum the divergent alternating PT expansion formally to a generalized Stieltjes integral $S(\beta)$,
\begin{equation}\label{biyoko}
\frac{S(\beta)}{\beta} = \int_{0}^{\infty}\frac{\rho(x)}{1+\beta x^2}\mathrm{d}x\sim \sum_{k=0}^{\infty}\mu_{2k} (-\beta)^{k}, \qquad \beta\to 0.
\end{equation}

The main impediment  with this approach is that evaluating $S(\beta)$ in the opposite asymptotic regime by expanding $(1+\beta x^2)^{-1}$ in powers of $1/(\beta x^2)$ and a term-wise integration is carried out leads to an expansion with terms that diverge individually,
\begin{align}\label{bijt}
S(\beta) = \beta\int_{0}^{\infty}\frac{\rho(x)}{1+\beta x^2}\mathrm{d}x \rightarrow \sum_{k=0}^{\infty} 
\mu _{-(2k+2)}\frac{(-1)^k}{\beta^k} ,
\end{align}
where the negative-power moments, $\mu_{-(2k+2)} =\int_{0}^{\infty}\rho(x)/x^{2k+2}\mathrm{d}x, k=0,1\dots $ are divergent integrals for sufficiently large $k$. For this reason, the evaluation of the Stieltjes integral after the reconstruction of $\rho(x)$ from the positive-power moments, $\mu_k$, is typically carried out numerically \cite{bender1987maximum, mead1984maximum}. We circumvent this problem by interpreting the divergent negative-power moments, $\mu_{-(2k+2)}$, as Hadamards's finite part and use its novel complex contour integral representation to arrive at a convergent expansion in inverse powers of $\beta$ plus a correction term that enables us to incorporate the logarithmic leading-order behavior in the strong-field regime.    

Finally in section \ref{conclusion}, we summarize our results and provide a possible direction in which to improve the efficacy of the summation and extrapolation prescription we presented here. In particular, we give a concise discussion on how to apply the prescription to extrapolate the divergent expansion \eqref{mopy} to the case when $\beta<0$, such as when the background is purely electric and the Heisenberg-Euler Lagrangian becomes complex-valued. The PT expansion becomes non-alternating in this case and our prescription also recovers the non-pertubative imaginary part of the exact result from the real coefficients of the non-alternating weak-field perturbation expansion. We emphasize that as opposed to a resurgent extrapolation to complex values of the parameter $\beta$ \cite{costin2019resurgent, costin2}, the extrapolation procedure we propose here for the weak-field PT expansion of the Heisenberg-Euler Lagrangian to the strong field regime is carried out along the real line so that the field values are physical. In which case, the extrapolation to the strong field limit in the direction of the positive real-axis, $\beta\to +\infty$, corresponds to an extrapolation to large magnetic fields while the negative direction, $\beta\to-\infty$, corresponds to an extrapolation to the strong electric field limit. 

\section{Hadamard's Finite Part}\label{sorat}
In this section, we give a concise discussion on the computation of Hadamard's finite part of the divergent integrals with a pole singularity at the origin,
 \begin{equation}\label{miv}
    \int_{0}^{a} \frac{f(x)}{x^{m}} \mathrm{d}x,\qquad m=1,2,\dots, a>0,
\end{equation}
where $f(x)$ is analytic at the origin and $f(0)\neq 0$. Consistent with the expedient canonical definition \cite{monegato2009definitions},  the Hadamard's finite part may be formulated more rigorously as a complex contour integral \cite{galapon2,galapon2016cauchy} or alternatively as a regularized limit at the poles of Mellin transform integrals \cite{regularizedlimit}. The equivalence of these dual representations is central to the results we 
 present here and their applications. 
 
\subsection{Canonical Representation}\label{hfp}
The canonical representation of the finite part of the divergent integral \eqref{miv} is obtained by introducing an arbitrarily small cut-off parameter $\epsilon$, $0<\epsilon<a$, to replace the offending non-integrable origin. The resulting convergent integral is grouped into two sets of terms
\begin{equation}\label{definitepart}
\int_{\epsilon}^{a} \frac{f(x)}{x^{m}} \mathrm{d}x = C_{\epsilon}+D_{\epsilon},
\end{equation}
where $C_{\epsilon}$ is the group of terms that possesses a finite limit as $\epsilon\rightarrow 0$, while $D_{\epsilon}$ diverges in the same limit and consists of terms in inverse powers of $\epsilon$ and $\ln\epsilon$. The finite part of the divergent integral is then defined uniquely by dropping the diverging group of terms $D_{\epsilon}$, leaving only the limit of $C_{\epsilon}$ and assigning the limit as the value of the divergent integral,
\begin{align}\label{finitepart}
\bbint{0}{a} \frac{f(x)}{x^{m}}\mathrm{d} x = \lim_{\epsilon\rightarrow 0} C_{\epsilon} .
\end{align}
The upper limit $a$ can be also be taken to infinity provided  $f(x)x^{-m}$ is integrable at infinity, in which case,
\begin{align}\label{pisik}
    \bbint{0}{\infty}\frac{f(x)}{x^{m}}\mathrm{d}x =     \lim_{a\to\infty}\bbint{0}{a}\frac{f(x)}{x^{m}}\mathrm{d}x.
\end{align}
\subsection{Contour Integral Representation}
A rigorous formulation of the Hadamard's finite part as a complex contour integral may also be derived from the following form of equation \eqref{definitepart} 
\begin{align}\label{form2}
\bbint{0}{a} \frac{f(x)}{x^{m}} \mathrm{d}x = \lim_{\epsilon\rightarrow 0} \left[\int_{\epsilon}^{a} \frac{f(x)}{x^{m}} \mathrm{d}x - D_{\epsilon}\right].
\end{align}
This contour integral representation is given in the following lemma.  The full derivation is given as a proof of Theorem 2.2 in \cite{galapon2}.
\begin{lemma}\label{prop1}
Let the complex extension, $f(z)$, of $f(x)$,  be analytic in the interval $[0,a]$. If $f(0)\neq 0$, then 
	\begin{equation}\label{result1}
	\bbint{0}{a}\frac{f(x)}{x^{m}}\mathrm{d}x=\frac{1}{2\pi i}\int_{\mathrm{C}} \frac{f(z)}{z^{m}} \left(\log z-\pi i\right)\mathrm{d}z, \;\; m = 1, 2 \dots 
	\end{equation}
where $\log z$ is the complex logarithm whose branch cut is the positive real axis and $\mathrm{C}$ is the contour straddling the branch cut of $\log z$ starting from $a$ and ending at $a$ itself, as depicted in figure \ref{tear2}. The contour $\mathrm{C}$ does not enclose any pole of $f(z)$.
\end{lemma}

\subsection{Regularized Limit Representation} \label{bilok}
Another equivalent formulation of the Hadamard's finite part is facilitated by the concept of the regularized limit introduced in \cite{regularizedlimit}. 
We quote here some of the relevant results.

Let $w(z)$ be a function of the complex variable $z$ that is analytic at some domain $D \subseteq \mathbb{C}$ and  $z_0$ be an isolated singularity of $w(z)$ in $D$, then we can define the deleted neighborhood $\delta_{z_0} = D(r,z_0)\setminus z_0$  where $D(r,z_0)$ is an open disk of radius $r$ centered at $z_0$ such that the function $w(z)$ admits the Laurent series expansion 
 \begin{align}\label{gisit}
     w(z) = \sum_{n=-\infty}^{\infty}a_n(z-z_0)^n,
 \end{align}
 where the coefficients $a_n$ are given by 
 \begin{align}
     a_n = \frac{1}{2\pi i} \oint_{ |z-z_0| = r'} \frac{w(z)}{(z-z_0)^{n+1}}\mathrm{d}z
 \end{align}
 for any $r' < r$.  The radius $r$ is bounded by the distance of the nearest singularity of $w(z)$ from $z_0$. The regularized limit of the function $w(z)$ at $z=z_0$ is defined as follows.
 
\begin{definition}
Let $z_0$ be an interior point in the domain D of $w(z)$. The regularized limit of $w(z)$ as $z\to z_0$, denoted by
\begin{equation}
        \lim^{\times}_{z\to z_0} w(z),
\end{equation}
is the coefficient $a_0$ in the Laurent series expansion \eqref{gisit} of $w(z)$ in a deleted neighborhood of $z_0$.
\end{definition}
In the case when the function $w(z)$ can be rationalized, that is written in the form $w(z) = h(z)/g(z)$ where $h(z)$
and $g(z)$ are both analytic at $z_0$ and $z_0$ is a simple zero of $g(z)$ while $h(z_0)\neq 0$, then the regularized limit is computed as,
\begin{align}\label{gilik}
    \lim_{z\to z_0}^{\times} \frac{h(z)}{g(z)} = \frac{h'(z_0)}{g'(z_0)} - \frac{h(z_0) g''(z_0))}{2(g'(z_0))^2}.
\end{align}
For the case when $g''(z_0) = 0$, this result reduces to a form similar to the L'Hospital's rule, 
\begin{align}\label{hospital}
    \lim_{z\to z_0}^{\times} \frac{h(z)}{g(z)} = \lim_{z\to z_0}\frac{h'(z)}{g'(z)}.
\end{align}
The derivation of this result is given in the proof of Corollary 3.1 in \cite{regularizedlimit}. 

The finite part of the divergent integral \eqref{miv} 
can be extracted from the 
analytic continuation, $\mathcal{M}^*[f(x); s]$, to the whole complex $s$ plane of the Mellin transform, 
\begin{align}\label{mopt}
    \mathcal{M} \left[f(x) ;\,\,s\right] = \int_{0}^{\infty}x^{s-1}f(x)\mathrm{d}x,
\end{align}
provided there is a non-trivial strip of analyticity of the Mellin integral. The case when the upper limit is a finite $a$ follows from equation \eqref{mopt} by considering $f(x) = g(x)\Theta(a-x)$, where $\Theta(x)$ is the Heaviside step function.  Furthermore, if $f(x)$ is analytic at $x=0$, the Mellin transform $\mathcal{M}[f(x); s]$ has at most simple poles along the real line. 

For a positive integer $m$, if $s=1-m$ is a pole of $\mathcal{M}^*[f(x); s]$, then the finite part integral \eqref{miv} is given by the regularized limit of $\mathcal{M}^*[f(x); s]$ at $s=1-m$,
\begin{align}\label{sigaa}
    \bbint{0}{\infty} \frac{f(x)}{x^{m}} \, \mathrm{d}x = \lim^{\times}_{s\to 1-m} \mathcal{M}^*\ \left[f(x) ;\,\,s\right].
\end{align}
In most situations, the analytic continuation can be written in rational form, $\mathcal{M}^*[f(x); s]=h(s)/g(s)$, where $h(1-m)\neq 0$ and $g(1-m)=0$. In this case, it may be possible to choose the rationalization such that $g''(1-m)=0$ so that when $s=1-m$ is a simple pole of $\mathcal{M}^*[f(x); s]=h(s)/g(s)$ the calculation of the regularized limit at $s=1-m$ reduces to equation \eqref{hospital}.

Both the canonical and regularized limit representations are used primarily for computing the value of the Hadamard's finite part explicitly. We demonstrate the foregoing discussion by computing the finite part integral \eqref{sigaa} for the case $f(x) = e^{-bx}$ for $b>0$. The finite part integral is the regularized limit at the poles of the analytic continuation of the following Mellin transform integral \cite[p 20]{brychkov2018handbook}, 
\begin{align}
    \mathcal{M} \left[e^{-b x} ;\,\,s\right] &= \int_{0}^{\infty}x^{s-1}e^{-bx}\mathrm{d}x\\\label{cipt}
    &= b^{-s}\Gamma(s), \qquad \mathrm{Re} \,b, \mathrm{Re}\,s > 0.
\end{align}
The analytic continuation of the Mellin transform to the whole complex plane is simply the right-hand side of equation \eqref{cipt},
\begin{align}\label{sif}
    \mathcal{M}^*\left[e^{-b x} ;\,\,s\right] = b^{-s}\Gamma(s).
\end{align}
\begin{figure}
    \centering
    \includegraphics[scale=0.28]{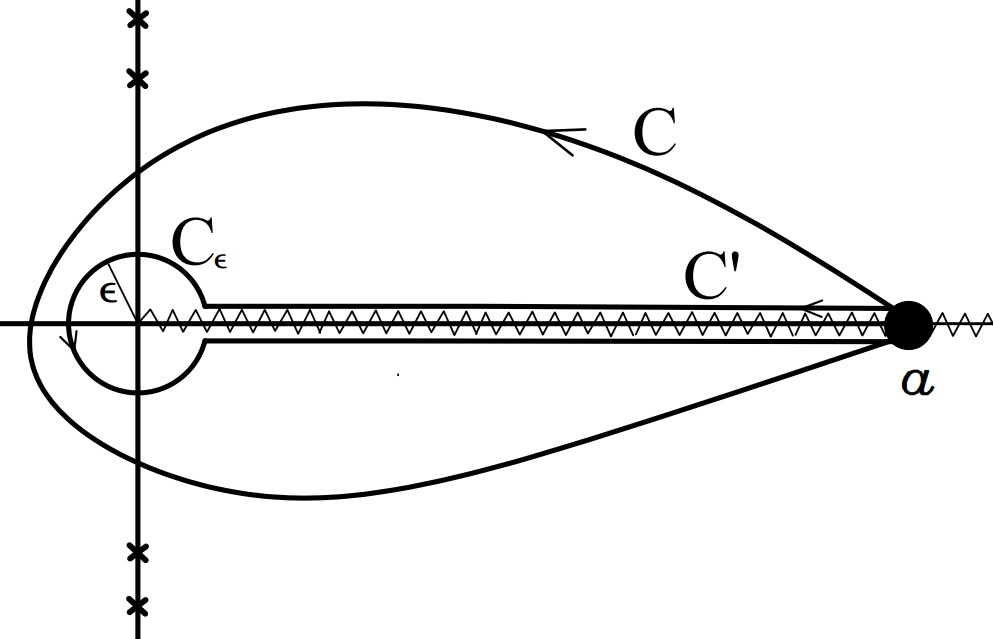}
	\caption{The contour $\mathrm{C}$ used in the representation \eqref{result1} for the Hadamard's finite part. The contour $\mathrm{C}$ excludes any of the poles of $f(z)$. The upper limit $a$ can be infinite. $\epsilon$ is a small positive parameter.}
	\label{tear2}
\end{figure}
Hence, the finite part integral is
\begin{align}\label{miwo}
        \bbint{0}{\infty}\,\frac{e^{-bx}}{x^{m}}\,\mathrm{d}x 
        =\lim^{\times}_{s\to 1-m} b^{-s}\Gamma(s)
        = \lim^{\times}_{s\to 1-m} \frac{\pi b^{-s}}{\sin(\pi s)\Gamma(1-s)}.
\end{align}
where we used the reflection formula $\Gamma(s) = \pi / \sin(\pi s)\Gamma(1-s)$. We then let $g(s) = \sin(\pi s)$ and $h(s) = \pi b^{-s}/ \Gamma(1-s)$ so that from equation \eqref{hospital}, the regularized limit evaluates to
\begin{align}\label{poriy}
    \bbint{0}{\infty}\,\frac{e^{-bx}}{x^{m}}\,\mathrm{d}x = \lim^{\times}_{s\to 1-m} \frac{h(s)}{g(s)} = \lim_{s\to 1-m} \frac{h'(s)}{g'(s)}.
\end{align}
Hence substituting $g'(s) = \pi\cos(\pi s)$ and 
\begin{align}
    h'(s) = -\frac{\pi b^{-s}}{\Gamma(1-s)}(\log b-\psi(1-s)),
\end{align}
where $\psi\left(z\right)$ is the digamma function and computing the limit in \eqref{poriy}, we obtain the result
\begin{equation}\label{ighi}
    \bbint{0}{\infty}\,\frac{e^{-bx}}{x^{m}}\,\mathrm{d}x = \frac{\left(-1\right)^{m}\,b^{m-1}}{\left(m-1\right)!}\left(\ln\,b-\psi\left(m\right)\right),\qquad m =1,2\dots,
\end{equation}
for real and positive $b$. A different rationalization of equation \eqref{sif} can be done but this will lead to the same result.

The result \eqref{ighi} is consistent with the result in \cite[eq 3.15]{galapon2} obtained after a lengthy calculation using the canonical definition \eqref{finitepart}. In this particular case, the computation of the finite part integral \eqref{ighi} as a regularized limit is more convenient. In general however, especially for cases when the Mellin transform $\mathcal{M}[f(x); s]$ does not exist for some function $f(x)$, the Hadamard's finite part integral can always be obtained from canonical definition \eqref{finitepart}.

\section{Finite-Part Integration}\label{meth}
An important application of the contour integral representation \eqref{result1} is a technique known as finite-part integration \cite{galapon2, tica2018finite,tica2019finite,villanueva2021finite} and relies on the equivalence among the three different representations discussed above. We apply it here to give a rigorous derivation of the closed-form for the Heisenberg-Euler Lagrangian from  which exact values can be computed. This is a special case of the more general result \cite[eq 64]{galapon3}.  In the next section, we will use these values to check the results of a summation and extrapolation procedure applied on the divergent weak-field expansions of the Heisenberg-Euler Lagrangian.  

\subsection{The Heisenberg-Euler Lagrangian} The Heisenberg-Euler Lagrangian $\mathcal{L}$ is a nonlinear correction to the Maxwell Lagrangian resulting from the interaction of a vacuum of charged particles of mass $m$ and spin $s$with an external electromagnetic field \cite{schwinger, dunne, dunne_harris,  heisenber_euler, dittrich, walter, walter2}. In the one-loop order and for the case of constant fields, it depends on the invariant quantity $\beta = e^2 (B^2-E^2)/m^4$, where $e$ is the electron charge. For a purely magnetic background, it is given for spin-$0$ and spin-$\frac{1}{2}$ particles as,  
\begin{align}\label{oin}
    \mathcal{L}_{0}(\beta) = \frac{m^4}{16\pi^2}f_{0}(\beta) \qquad\text{and}\qquad \mathcal{L}_{\frac{1}{2}}(\beta) = \frac{m^4}{8\pi^2}f_{\frac{1}{2}}(\beta) 
\end{align}
where,
\[
       f_s(\beta) = \int_{0}^{\infty}\frac{\mathrm{d}\tau}{\tau^3}e^{-\tau}\chi_s{(\sqrt{\beta}\, \tau)};\qquad\chi_s(x) = 
\begin{dcases}
    x\,\mathrm{csch}(x) - 1+\frac{x^2}{6}, &  s = 0 \\
    1+\frac{x^2}{3}-x\coth(x), &    s = \frac{1}{2} . 
\end{dcases}
\]
These representations are in Heaviside-Lorentz system with natural units so that $h = c =1$ and the fine structure constant reads $\alpha = e^2/4\pi$.

For $x\to 0$, the function $\chi_s(x)$ can be expanded as,
\begin{align}\label{miyp}
    \chi_s(x) = \sum_{k=2}^{\infty} c^{(s)}_k x^{2k},
\end{align}
where the expansion coefficients are given by 
\begin{align}
    c_k^{(0)} = \frac{2-2^{2k}}{(2k)!}B_{2k},\quad c_k^{(1/2)} =-\frac{2^{2k}}{(2k)!} B_{2k}
\end{align}
and $B_{2k}$ are the Bernoulli numbers. Substituting this expansion for $\chi_s(x)$ into equation \eqref{oin} and integrating term-by-term, we obtain the divergent PT expansion for the function $f_s(\beta)$: 
\begin{equation}\label{gagah}
    f_{s}(\beta) = \sum_{k=2}^{\infty} a_k^{(s)} (-\beta)^{k},\qquad a_{k}^{(s)} = (-1)^{k}(2k-3)! c_k^{(s)}, \qquad \beta\to 0.
\end{equation}

We consider exclusively the case of a purely magnetic background where $f_s(\beta)$ in the integral representation \eqref{oin} is real and the expansion \eqref{gagah} is alternating in sign. The Heisenberg-Euler Lagrangian is complex-valued for the case of a purely electric background,  $\beta = - \kappa, \kappa = e^2 E^2/m^4 $. The imaginary part signals the important phenomenon of Schwinger effect which predicts the instability of the QED vacuum \cite{dunne_shw}. The imaginary part is invisible up to any finite order of the PT expansion \eqref{gagah} which in this case becomes non-alternating in sign. The procedure we present in the following subsections works equally well for this case. We will give an outline of the procedure when applied to this case in section \ref{conclusion} while a detailed treatment is given in \cite{nonalternating}. 

While $f_s(\beta)$ in equation \eqref{oin} is well-defined in both spin cases, it is expressed as a finite sum of divergent integrals \eqref{miv} with pole singularities at the origin. A formal procedure to arrive at a closed-form is carried out in \cite{dittrich, walter, walter2} by employing $n$-dimensional regularization scheme to make sense of the divergent integrals. Here, we will make use of the equivalent representations of Hadamard's finite part to derive these results more rigorously. 

\subsection{Spin-0}
For the Heisenberg-Euler Lagrangian \eqref{oin} in the case of the spin-0 particles, consider the contour integral,
\begin{align}
    \int_{\mathrm{C}} \frac{g(z)}{z^3}\log z\,\mathrm{d}z,\qquad 
    g(z) = e^{-z} \left(\sqrt{\beta}\,z\, \mathrm{csch} \left(\sqrt{\beta} z\right) - 1 +\frac{\beta z^2}{6}\right),
\end{align}
where the contour $\mathrm{C}$ is given in figure \ref{tear2}. None of the poles of the integrand along the imaginary axis are in the interior of $\mathrm{C}$. We deform $\mathrm{C}$ to an equivalent contour $\mathrm{C'}$ so that 
\begin{align}\label{nimo}
    \int_{0}^{a} \frac{g(x)}{x^3} \mathrm{d}x = \frac{1}{2\pi i}\int_{\mathrm{C}}\frac{g(z)}{z^3} \log z \mathrm{d}z,
\end{align}
where the integral along the circular contour about the origin, $\mathrm{C}_\epsilon$, vanishes as $\epsilon\to 0$. We then add a zero term,
\begin{align}
    -\frac{\pi i}{2\pi i}\int_{\mathrm{C}} \frac{g(z)}{z^3} \mathrm{d}z = -\pi i \sum \mathrm{Res}\,\,\left[\frac{g(z)}{z^3}\right] = 0,
\end{align}
to the right-hand side of equation \eqref{nimo} so that,
\begin{align}
    \int_{0}^{a} \frac{g(x)}{x^3} \mathrm{d}x = \frac{1}{2\pi i}\int_{\mathrm{C}}\frac{g(z)}{z^3}\left(\log z - \pi i \right)\mathrm{d}z
    = \bbint{0}{a}  \frac{g(x)}{x^3}\mathrm{d}x,
\end{align}
where we used the contour integral representation \eqref{result1} of the Hadamard's finite part integral. Hence, taking the limit $a\to\infty$, the exact integral representation \eqref{oin} evaluates to 
\begin{align}\nonumber
     f_0(\beta) &= \int_{0}^{\infty}\frac{e^{-\tau}}{\tau^3}\left[\sqrt{\beta}\,\tau \,\mathrm{csch}{\left(\sqrt{\beta}\tau\right)} - 1 + \frac{\beta\tau^2}{6}\right] \mathrm{d}\tau\\\label{gidak}
     &= \sqrt{\beta}\bbint{0}{\infty}\frac{e^{-\tau} \mathrm{csch}{\left(\sqrt{\beta}\tau\right)}}{\tau^2}\mathrm{d}\tau -\bbint{0}{\infty}\frac{e^{-\tau}}{\tau^3}\mathrm{d}\tau + \frac{\beta}{6}\bbint{0}{\infty} \frac{e^{-\tau}}{\tau}\mathrm{d}\tau.
\end{align}
In effect, the result \eqref{gidak} follows directly from a term-by-term integration followed by the immediate regularization of each divergent term as Hadamard's finite part. In section \ref{bigaj}, we will demonstrate that this procedure will generally result to missing terms especially when one performs a term-by-term integration involving an infinite number of divergent integrals. 

The first finite part integral in  right-hand side of equation \eqref{gidak} can be computed using the following Mellin transform integral \cite[p 34, eq 7]{brychkov2018handbook},
\begin{align}\nonumber
    \mathcal{M}\left[e^{-a\tau} \mathrm{csch}{\left(b\,\tau\right)};\,\,s\right] &= \int_{0}^{\infty}\tau^{s-1} e^{-a\tau} \mathrm{csch}{\left(b\,\tau\right)} \mathrm{d}\tau\\\label{igik}
    & = \frac{2^{1-s}}{b^{s}}\Gamma(s)\,\zeta\left(s,\frac{a+b}{2\,b}\right), \qquad \mathrm{Re}\,a > -|\mathrm{Re}\,b|; \mathrm{Re}\,s > 1,
\end{align}
 where $\,\zeta\left(z,\nu\right)$ is the Hurwitz zeta function. The finite part integral is the regularized limit at the simple pole $s=-1$ of the analytic continuation of the Mellin transform to the whole complex plane. The analytic continuation of the Mellin transform is simply given  by the right-hand side of equation \eqref{igik},
 \begin{align}
\mathcal{M}^{\ast}\left[e^{-a\tau} \mathrm{csch}{\left(b\,\tau\right)};\,\,s\right] = \frac{2^{1-s}}{b^{s}}\Gamma(s)\,\zeta\left(s,\frac{a+b}{2\,b}\right).
 \end{align}
 So that
\begin{align}
    \bbint{0}{\infty}\frac{e^{-\tau} \mathrm{csch}{\left(\sqrt{\beta}\tau\right)}}{\tau^2}\mathrm{d}\tau 
    &= \lim_{s\to-1}^{\times} \mathcal{M}^{\ast}\left[e^{-\tau} \mathrm{csch}{\left(\sqrt{\beta}\,\tau\right)};\,\,s\right] \\
    &= \lim_{s\to-1}^{\times} \frac{2^{1-s}}{\beta^{s/2}} \frac{\pi}{\sin(\pi s)\Gamma(1-s)}\,\zeta\left(s,\frac{1+\sqrt{\beta}}{2\sqrt{\beta}}\right).
\end{align}
where we've made use of the reflection formula for $\Gamma(s)$.
We then rationalize the analytic continuation of the Mellin transform by writing,
\begin{align}
\mathcal{M}^{\ast}\left[e^{-\tau} \mathrm{csch}{\left(\sqrt{\beta}\,\tau\right)};\,\,s\right] = \frac{h(s)}{g(s)},
\end{align}
where
\begin{equation}\label{igli}
    h(s) = \frac{2^{1-s}\pi}{\beta^{s/2}\Gamma(1-s)} \zeta\left(s,\frac{1+\sqrt{\beta}}{2\sqrt{\beta}}\right), \qquad g(s) = \sin(\pi s).
\end{equation}
From the formula \eqref{hospital}, we compute the regularized limit as
\begin{align}\label{miytp}
      \bbint{0}{\infty}\frac{e^{-\tau} \mathrm{csch}{\left(\sqrt{\beta}\tau\right)}}{\tau^2}\mathrm{d}\tau = \lim_{s\to-1}^{\times} \frac{h(s)}{g(s)} = \lim_{s\to-1} \frac{h'(s)}{g'(s)}.
\end{align}
Computing the derivatives, $g'(s) = \pi\cos(\pi s)$ and 
\begin{align}\nonumber
    h'(s) = \frac{2^{1-s}\pi\beta^{-s/2}}{\Gamma(1-s)}& \left[\left(\psi(1-s)-\ln\beta-\ln 2\right)\,\zeta\left(s,\frac{1+\sqrt{\beta}}{2\sqrt{\beta}}\right)\right.\\
    + &\left.\,\zeta^{(1,0)}\left(s,\frac{1+\sqrt{\beta}}{2\sqrt{\beta}}\right)  \right],
\end{align}
where $\zeta^{(1,0)}(z,\nu)$ is the derivative of the Hurwitz zeta function with respect to the first argument $z$. Hence, the finite part integral \eqref{miytp} evaluates to 
\begin{align}\nonumber
          \bbint{0}{\infty}\frac{e^{-\tau} \mathrm{csch}{\left(\sqrt{\beta}\tau\right)}}{\tau^2}\mathrm{d}\tau = 2\sqrt{\beta} &\left[\left(\ln\beta+\ln 4+2\gamma-2\right)\,\zeta\left(-1,\frac{1+\sqrt{\beta}}{2\sqrt{\beta}}\right) \right.\\
          &- \left. 2\,\zeta^{(1,0)}\left(-1,\frac{1+\sqrt{\beta}}{2\sqrt{\beta}}\right)\right],
\end{align}
where $\gamma = - \psi(1)$ is the Euler-Mascheroni constant. 

The other finite part integrals in the right-hand side of equation \eqref{gidak} are special cases of equation \eqref{ighi}. Hence, the integral representation \eqref{gidak} for the Heisenberg-Euler Lagrangian in the case of spin-0 particles takes the following closed-form,
\begin{equation}\label{impin}
f_0(\beta) = \frac{\beta\ln\beta}{12} - \frac{\ln\beta}{4}+\beta\left(\frac{\ln 4}{12}-\frac{1}{6}\right) -\frac{\ln 4}{4}-\frac{1}{4}- 4\beta \zeta^{(1,0)}\left(-1,\frac{1+\sqrt{\beta}}{2\sqrt{\beta}}\right).
\end{equation}
where we made use of the relation \cite{NIST},
\begin{equation}\label{hilaga}
   \zeta(-1, \nu) = -\frac{1}{12}+\frac{\nu}{2}-\frac{\nu^2}{2},
\end{equation}
to simplify the result.
The result \eqref{impin} is consistent with that given in \cite{dunne,walter,walter2} using a different approach based on $\zeta$-function regularization. The result given in \cite[eq 3.26]{dittrich} obtained formally using $n$-dimensional regularization has an error and is rectified in a subsequent work \cite{walter}. 

Similarly for the case of spin-$\frac{1}{2}$ particles,
\begin{align}\label{kity}
    f_\frac{1}{2}(\beta) = \bbint{0}{\infty}\frac{e^{-\tau}}{\tau^{3}}\mathrm{d}\tau+\frac{\beta}{3}\bbint{0}{\infty}\frac{e^{-\tau}}{\tau}\mathrm{d}\tau-\sqrt{\beta} \bbint{0}{\infty}\frac{e^{-\tau}\coth{\left(\sqrt{\beta}\tau\right)}}{\tau^2}\mathrm{d}\tau.
\end{align}
The first two finite part integrals are again special cases of the result \eqref{ighi} while the finite part integral in the third term is the regularized limit at a pole of the analytic continuation of the Mellin transform \cite[p 34, eq 11]{brychkov2018handbook} to the whole complex plane,
\begin{align}\nonumber
    \bbint{0}{\infty}&\frac{e^{-\tau}\coth{\left(\sqrt{\beta}\tau\right)}}{\tau^2}\mathrm{d}\tau = \sqrt{\beta}\left(\ln 16 + 2\ln\beta \right)\zeta\left(-1,\frac{1}{2\sqrt{\beta}}\right)\\\label{imaw}
    &+ (\gamma-1)\left(4\sqrt{\beta}\zeta\left(-1, \frac{1}{2\sqrt{\beta}}\right)-1\right)
    - 4\sqrt{\beta}\zeta^{(1,0)}\left(-1, \frac{1}{2\sqrt{\beta}}\right).
\end{align}
Substituting the finite part \eqref{imaw} to equation \eqref{kity} and making use of the relation \eqref{hilaga}, we arrive at the closed-form, 
\begin{align}\nonumber
        f_{\frac{1}{2}}(\beta) = 4\beta\,\zeta^{(1,0)}\left(-1, \frac{1}{2\sqrt{\beta}}\right)& + \frac{1}{4} - \frac{\beta}{3}\\\label{kilat}
        &-\beta\left(\ln 16 + 2\ln\beta \right)\,\left(-\frac{1}{12}+\frac{1}{4\sqrt{\beta}}-\frac{1}{8\beta}\right).
\end{align}
This result also coincides with the result reported in \cite{dittrich, walter} obtained from the integral representation \eqref{oin} using $n$-dimensional regularization; and in \cite{dunne,walter2} using $\zeta$-function regularization. 

\subsection{Self-dual electromagnetic background.} In the case of a self-dual (SD) electromagnetic background \cite{sara,dunne, honda, self_dual1, self_dual2} satisfying 
\begin{equation}
    F_{\mu\nu} = \widetilde{ F}_{\mu\nu} \equiv \frac{1}{2}\epsilon_{\mu\nu\rho\lambda}F^{\rho\lambda},
\end{equation}
the one-loop Heisenberg-Euler Lagrangian describing a charged scalar particle is given by 
\begin{equation}\label{muska}
    \mathcal{L}_0(\beta) = \frac{e^2 \mathcal{F}^2}{4\pi^2}f_{\mathrm{SD}}(\beta), \,\,\, f_{\mathrm{SD}}(\beta) = \frac{1}{4}\int_{0}^{\infty}\frac{\mathrm{d}\tau}{\tau}\left(\frac{1}{\sinh^2{\tau}}-\frac{1}{\tau^2}+\frac{1}{3}\right)e^{-{2\tau}/{\sqrt{\beta}}},
\end{equation}
where $\mathcal{F}^2 = \frac{1}{4}F_{\mu\nu}F^{\mu\nu}$ and the natural dimensionless parameter $\beta = \left(2e\mathcal{F}/m^2\right)^2$. The same integral, $f_{\mathrm{SD}}(\beta)$, also appears in the case of spinor QED.
This integral is also relevant in the non-perturbative formulation for the free energy in the $c=1$ string theory at self-dual radius \cite{sara}. We also limit to the case when $\mathcal{F}$ is real or magnetic-like. Our treatment here to arrive at a closed-form  applies equally well for an imaginary or electric-like background, $\mathcal{F} = -i \overline{\mathcal{F}}$. The details are also given in \cite{nonalternating}. 

As in the previous cases, finite-part integration applied on the exact integral formulation \eqref{muska} yields,  
\begin{align}\label{futi}
   f_{\mathrm{SD}}(\beta)  = \frac{1}{4}\bbint{0}{\infty}\frac{e^{-{2\tau}/{\sqrt{\beta}}}}{\tau \sinh^2{\tau}}\mathrm{d}\tau -\frac{1}{4}\bbint{0}{\infty}\frac{e^{-{2\tau}/{\sqrt{\beta}}}}{ {\tau^3}}\mathrm{d}\tau +\frac{1}{12}\bbint{0}{\infty}\frac{e^{-{2\tau}/{\sqrt{\beta}}}}{ {\tau}}\mathrm{d}\tau.
\end{align}
The finite part integral in the first term in the right-hand side of equation \eqref{futi} is computed as a regularized limit at a pole of the Mellin transform integral \cite[p 29, eq 10]{brychkov2018handbook}.
 Proceeding similarly as in the previous cases, we obtain
\begin{align}\nonumber
    \frac{1}{4}\bbint{0}{\infty}\frac{e^{-{2\tau}/{\sqrt{\beta}}}}{\tau \sinh^2{\tau}}\mathrm{d}\tau =  \left(-\gamma-\ln 2\right)& \left[\zeta\left(-1,\frac{1}{\sqrt{\beta}}\right) -\frac{1}{\sqrt{\beta}}\zeta\left(0,\frac{1}{\sqrt{\beta}}\right)\right] \\
     &+ \zeta^{(1,0)}\left(-1,\frac{1}{\sqrt{\beta}}\right) -  \frac{1}{\sqrt{\beta}}\,\zeta^{(1,0)}\left(0,\frac{1}{\sqrt{\beta}}\right).
\end{align}
The other  finite part integrals occurring in \eqref{futi} are again given by equation \eqref{ighi}. 
We then make use of $\zeta(0, \nu) = \frac{1}{2}-\nu$ along with the relation \eqref{hilaga} to write,
\begin{align}\label{nopita}
    f_{\mathrm{SD}}(\beta) = \zeta^{(1,0)}\left(-1,\frac{1}{\sqrt{\beta}}\right) - \frac{1}{\sqrt{\beta}}\,\zeta^{(1,0)}\left(0,\frac{1}{\sqrt{\beta}}\right) 
    -(\ln{\beta})\left(\frac{1}{4\beta}-\frac{1}{24}\right)-\frac{3}{4\beta}.
\end{align}
A result given in \cite[eq 2.5 p 5]{self_dual1} writes $f_{\mathrm{SD}}(\beta)$ in terms of the Barnes G function. 

\begin{table}
	\begin{tabular}{ l l l l l l }
		\hline
		$d$   &  $\beta=10^{-2}$ & $\beta=0.1$&  $\beta=0.2$\\ 
		\hline
		1   & $\textcolor{blue}{1.932}143(10^{-6})$ & $\textcolor{blue}{1.8}214(10^{-4})$ & $6.7937(10^{-4})$ \\
            2   & $\textcolor{blue}{1.932394}841(10^{-6})$  & $\textcolor{blue}{1.8}466(10^{-4})$ & $7.1968(10^{-4})$\\
		5   & $\textcolor{blue}{1.9323847}854(10^{-6})$ & $\textcolor{blue}{1.83}50(10^{-4})$ &$6.235({10^{-4}})$\\
		9   & $\textcolor{blue}{1.932384796}847(10^{-6})$ & $\textcolor{blue}1.6194(10^{-4})$ 
                &$-4.978(10^2)$ \\
		20  & $\textcolor{blue}{1.9323847969}843(10^{-6})$ & $8.42618(10^{5})$ & $3.636(10^{12})$\\
		50  & $3.3995123(10^4)$ &  &\\
          \hline
		Exact & $ 1.93238479692775525(10^{-6})$ & $ 1.83994677220(10^{-4})$ & $7.0356826048(10^{-4})$ \\
		\hline
  
	\end{tabular}
 
	\caption{Convergence of the partial sums of the perturbative expansion \eqref{gagah} for the integral $f_0(\beta)$ in equation \eqref{oin} for the Heisenberg-Euler Lagrangian in the case of spin-0 particles. The exact result is computed from the closed-form \eqref{impin} of the integral representation \eqref{oin}.}
 \label{bigak}
\end{table} 

\section{Summation and extrapolation of the weak-field expansion for the Heisenberg-Euler Lagrangian}\label{bigaj}
We now discuss the main result of this paper which is the application of the method of finite-part integration in the important problem of summation and extrapolation to the strong-field regime of the divergent weak-field expansion given in equation \eqref{gagah} for the Heisenberg-Euler Lagrangian. 
The divergence of the PT expansion \eqref{gagah} is evident in the leading growth rate of the coefficients, $a_k^{(s)}\sim(2k)!$ as $k\to\infty$ \cite{dunne}, and from the results presented in table \ref{bigak} of using partial sums of this expansion to compute $f_0(\beta)$ for some values of the parameter $\beta$. In addition,
the function $f_s(\beta)$ possesses the leading-order behavior in the strong-field regime, $\beta\to\infty$ \cite[eqs 1.63 and 1.54]{dunne},
\begin{align}\label{mirt}
    f_{0}(\beta)\sim\frac{\beta\ln\beta}{12} +\frac{\ln 2}{6}\beta + \dots\,\,\, \text{and}\,\,\, f_{\frac{1}{2}}(\beta) \sim \frac{\beta\ln \beta }{6} + \frac{\ln 2}{3}\beta + \dots 
\end{align}
 
On the basis of this information, we sum the divergent PT series \eqref{gagah}
by mapping the first $d+1$ expansion coefficients $a_{k+2}^{(s)}$ to the positive-power moments $\mu_{2k}^{(s)}$ of some positive function $\rho_s(x)$,
\begin{align}\label{gigh}
    a_{k+2}^{(s)} =\mu_{2k}^{(s)} = \int_{0}^{\infty} x^{2k}\rho_s(x)\mathrm{d}x, \qquad k = 0,1,\dots, d.
\end{align}
Substituting this to the expansion \eqref{gagah}, the function $f_s(\beta)$ is summed formally as,
\begin{align}\nonumber
    f_s(\beta) &= \sum_{k=2}^{\infty} a_k^{(s)} (-\beta)^{k} = \beta^{2}\sum_{k=0}^{\infty} a_{k+2}^{(s)} (-\beta)^{k} = \beta^{2} \sum_{k=0}^{\infty} \int_{0}^{\infty} x^{2k} \rho_s(x) (-\beta)^{k} \mathrm{d}x\\
    &= \beta^{2}\int_{0}^{\infty}\rho_s(x) \left(\sum_{k=0}^{\infty} (-\beta x^2)^k\right)\mathrm{d}x = \beta^2 \int_{0}^{\infty} \frac{\rho_s(x)}{1+\beta x^{2}} \mathrm{d}x.
\end{align}
which is in terms of the generalized Stieltjes integral,
\begin{align}
    S(\beta) =  \int_{0}^{\infty} \frac{\rho_s(x)}{1/\beta + x^{2}} \mathrm{d}x.
\end{align}

The next step is to evaluate $S(\beta)$ in the strong-field regime in a manner that allows one to incorporate the known leading-order strong-field behavior \eqref{mirt}. The relevant expansion for the generalized Stieltjes integral is obtained using the method of finite-part integration. This is given in the following theorem.

\begin{theorem} \label{lemma0}
Let the complex extension, $\rho(z)$, of the real-valued function $\rho(x)$ for real $x$, be entire, then the generalized Stieltjes integral, $S(\beta)$, admits the following exact convergent expansion
        \begin{align}\label{som}
       S(\beta) = \int_{0}^{\infty}\frac{\rho(x)}{1/\beta + x^{2}}\mathrm{d}x = \sum_{k=0}^{\infty} 
        \frac{(-1)^k}{\beta^k} \mu_{-(2k+2)} +\Delta(\beta),
        \end{align}
        where the term $\Delta(\beta)$ is given by
\begin{align}\label{gibad}
    \Delta(\beta) = \frac{\pi \sqrt{\beta}}{4}\left(\rho\left(\frac{i}{\sqrt{\beta}}\right)+\rho\left(\frac{-i}{\sqrt{\beta}}\right)\right)
    +\frac{\sqrt{\beta}\ln\beta}{4 i}\left(\rho\left(\frac{i}{\sqrt{\beta}}\right)-\rho\left(\frac{-i}{\sqrt{\beta}}\right)\right),
\end{align}
and $\mu_{-(2k+2)}$ are the divergent negative-power moments of $\rho(x)$ are interpreted as the Hadamard's finite part integral,
\begin{equation}\label{pigil}
    \mu_{-(2k+2)} = \bbint{0}{\infty}\frac{\rho(x)}{x^{2k + 2}}\mathrm{d}x.
\end{equation}
\end{theorem}

\begin{proof}
   Deform the contour $\mathrm{C}$ to $\mathrm{C'}$ as shown in figure \ref{tear} and perform the following contour integration, 
\begin{equation}\label{tint}
\int_{\mathrm{C'}} \frac{\rho(z)}{1/\beta + z^{2}} \log z\,\mathrm{d}z = (2\pi i) \int_{0}^{a}\frac{\rho(x)}{1/\beta+ x^{2}}\mathrm{d}x + (2\pi i)\sum \mathrm{Res}\left[\frac{\rho(z)\log z}{1/\beta+z^{2}}\right]
\end{equation}
where the integral along the circular loop, $\mathrm{C_\epsilon}$, vanishes as $\epsilon\to 0$. This yields an expression for the original integral along the real line,
\begin{align}\label{mit}
    \int_{0}^{a}\frac{\rho(x)}{1/\beta+ x^{2}}\mathrm{d}x = \frac{1}{2\pi i} \int_{\mathrm{C}} \frac{\rho(z)}{1/\beta + z^{2}} \log z\,\mathrm{d}z  - \sum \mathrm{Res}\left[\frac{\rho(z)\log z}{1/\beta+z^{2}}\right].
\end{align}
We then add a zero to the first term of the right-hand side of equation \eqref{mit}, by adding and subtracting the term 
\begin{align}
    \frac{\pi i}{2\pi i}\int_{\mathrm{C}} \frac{\rho(z)}{1/\beta + z^{2}}\mathrm{d}z = \pi i\sum \mathrm{Res}\left[\frac{\rho(z)}{1/\beta+z^{2}}\right].
\end{align}
So that the first term in the right-hand side of equation \eqref{mit} becomes, 
\begin{align}\nonumber\label{musa}
     \frac{1}{2\pi i} \int_{\mathrm{C}} \frac{\rho(z)}{1/\beta + z^{2}} \log z\,\mathrm{d}z  
     = \frac{1}{2\pi i} \int_{\mathrm{C}} \frac{\rho(z)}{1/\beta + z^{2}} (\log z- \pi i)\,\mathrm{d}z  \\
     + \pi i\sum \mathrm{Res}\left[\frac{\rho(z)}{1/\beta+z^{2}}\right].
\end{align}

 \begin{figure}
    \centering
    \includegraphics[scale=0.21]{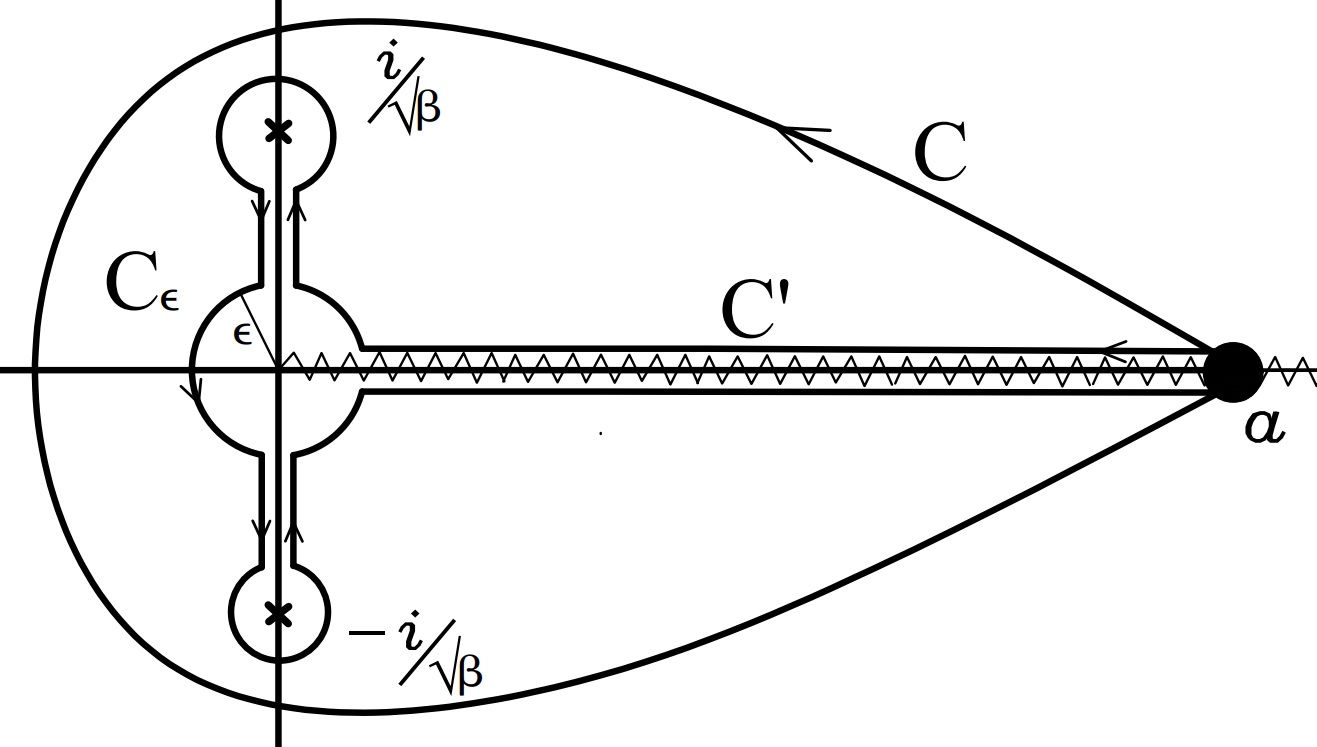}
	\caption{The contour of integration. The upper limit $a$ can be infinite. The poles of the generalized Stieltjes integral at $z=\pm i/\sqrt{\beta}$ lie inside the contour $\mathrm{C}$. $\epsilon$ is a small positive parameter.}
	\label{tear}
\end{figure}

 In the first term of the right-hand of equation \eqref{musa}, we expand $\left(1/\beta + z^{2}\right)^{-1}$ in powers of $1/(\beta z^2)$ and perform a term-by-term integration so that equation \eqref{musa} becomes
\begin{align}\nonumber\label{gabt}
         \frac{1}{2\pi i} \int_{\mathrm{C}} \frac{\rho(z)}{1/\beta + z^{2}} \log z\,\mathrm{d}z =& \sum_{k=0}^{\infty} 
        \frac{(-1)^k}{\beta^k}\frac{1}{2\pi i } \int_{\mathrm{C}}\frac{\rho(z)(\log z-\pi i)}{z^{2k+2}}\mathrm{d}z \\
        &+\pi i\sum \mathrm{Res}\left[\frac{\rho(z)}{1/\beta+z^{2}}\right],\qquad |z| > \frac{1}{\sqrt{\beta}}\\  \label{kinb}
        =&\sum_{k=0}^{\infty} 
        \frac{(-1)^k}{\beta^k} \bbint{0}{a}\frac{\rho(x)}{x^{2k + 2}}\mathrm{d}x + \pi i\sum \mathrm{Res}\left[\frac{\rho(z)}{1/\beta+z^{2}}\right],
\end{align}
where in equation \eqref{kinb}, we used the contour integral representation \eqref{result1} of the Hadamard's finite part integral.   

Substituting equation \eqref{kinb} for the first term of the right-hand side of equation \eqref{mit} and taking the limit as $a\to\infty$,
\begin{align}\label{alad}
\lim _{a\to\infty}\int_{0}^{a}\frac{\rho(x)}{1/\beta + x^{2}}\mathrm{d}x = \sum_{k=0}^{\infty} 
\frac{(-1)^k}{\beta^k} \lim_{a\to\infty}\bbint{0}{a}\frac{\rho(x)}{x^{2k + 2}}\mathrm{d}x +\Delta(\beta).
\end{align}
The term $\Delta(\beta)$ is given by
\begin{align}\label{igty}\nonumber
    \Delta(\beta) &= \sum \mathrm{Res}\left[\frac{\rho(z)(\log z - \pi i)}{1/\beta+z^{2}}; z =\pm \frac{i}{\sqrt{\beta}}\right]
\end{align}
which evaluates to equation \eqref{gibad}. Hence, we obtain the result \eqref{som}.

\end{proof}

\begin{table}
	\begin{tabular}{c lllllll}
		\hline
		Moments & $\beta = 10^{7}$ & $\beta = 10^{12}$ &   $\beta = 10^{13}$ & $\beta = 10^{18}$ \\ 
		\hline

		100  & $\textcolor{blue}{1.07}87(10^{7})$ & $\textcolor{blue}{2.0}424(10^{12})$  & $\textcolor{blue}{2.2}3516(10^{13})$ & $\textcolor{blue}{3.1}991(10^{18})$  \\
  
 	500 & $\textcolor{blue}{1.07}54(10^{7})$ & $\textcolor{blue}{2.03}27(10^{12})$ & $\textcolor{blue}{2.22}416(10^{13})$ & $\textcolor{blue}{3.18}16(10^{18})$  \\ 


		1000 & $\textcolor{blue}{1.076}3(10^{7})$  & $\textcolor{blue}{2.03}47(10^{12})$  & $\textcolor{blue}{2.22}648(10^{13})$  & $\textcolor{blue}{3.18}51(10^{18})$   \\ 

            1500 &  $\textcolor{blue}{1.07}72(10^{7})$  & $\textcolor{blue}{2.036}7(10^{12})$  & $\textcolor{blue}{2.228}60(10^{13})$ & $\textcolor{blue}{3.18}83(10^{18})$\\ 

            2000 & $\textcolor{blue}{1.07}71(10^{7})$ & $\textcolor{blue}{2.036}4(10^{12})$ & $\textcolor{blue}{2.228}33(10^{13})$ & $\textcolor{blue}{3.187}9(10^{18})$  \\
      
            2500 & $\textcolor{blue}{1.0769}4(10^{7})$ & $\textcolor{blue}{2.0361}5(10^{12})$ & $\textcolor{blue}{2.2280}3(10^{13})$ & $\textcolor{blue}{3.1874}5(10^{18})$ \\
        \hline
        $P^{999}_{1000} (\beta)$ & $1.5148(10^{4})$  & $1.5151(10^{9})$ & $1.5151(10^{12})$  & $1.5151(10^{15})$  \\
       
        $P^{49}_{50} (\beta)$ & $1.0723(10^{6})$  & $1.0723(10^{11})$ & $1.0723(10^{12})$  & $1.0723(10^{17})$ \\
        \hline
        $\delta_{499} ( \beta)$ & $8.5224(10^{6})$ & $1.0137(10^{15})$ & $1.0130(10^{17})$ & $1.0129(10^{27})$  \\
        
        $\delta_{100} (\beta)$ & $1.1943(10^{7})$ & $6.1881(10^{16})$ & $6.1880(10^{18})$ & $6.1880(10^{28})$  \\
        \hline
        Exact & $1.07693(10^{7})$ & $2.03613(10^{12})$ & $2.22801(10^{13})$ & $3.18742(10^{18})$ \\
        \hline

	\end{tabular}
 
	\begin{tabular}{c lllllll}

		Moments & $\beta = 1$ & $\beta = 4$ &  $\beta = 10^2$ & $\beta = 10^3$ & $\beta = 10^4$ &  \\ 
		\hline

		100 & $\textcolor{blue}{0.0139}583$ & $\textcolor{blue}{0.149}42$  & $\textcolor{blue}{17}.2803$ & $\textcolor{blue}{32}8.27$ & $\textcolor{blue}{507}1.2$ \\

 	500 & $\textcolor{blue}{0.0139688}5101 $ & $\textcolor{blue}{0.149783}78$ & $\textcolor{blue}{17.35}40$ & $\textcolor{blue}{329}.10$ & $\textcolor{blue}{507}3.5$  \\ 

		800 & $\textcolor{blue}{0.01396884}6277$ & $\textcolor{blue}{0.149783}10$ & $\textcolor{blue}{17.35}31$ & $\textcolor{blue}{329}.08$ & $\textcolor{blue}{507}3.0$\\ 

		1000 & $\textcolor{blue}{0.013968847}5625$ & $\textcolor{blue}{0.149783}54$  & $\textcolor{blue}{17.35}52$ & $\textcolor{blue}{329}.19$ & $\textcolor{blue}{507}5.9$ \\ 

            1500 & $\textcolor{blue}{0.0139688479}565$ & $\textcolor{blue}{0.1497837}32$ & $\textcolor{blue}{17.356}5$  & $\textcolor{blue}{329.2}68$ & $\textcolor{blue}{507}8.1$ \\ 

            2000 & $\textcolor{blue}{0.0139688479}511$ & $\textcolor{blue}{0.14978372}6$ & $\textcolor{blue}{17.356}4$ & $\textcolor{blue}{329.2}60$ & $\textcolor{blue}{5077}.9$\\

        \hline
        $P^{999}_{1000} (\beta)$ & $0.0139688428836$ & $0.149678652$ & $0.63315$ & $13.0395$ & $149.04$ \\
        
         $P^{49}_{50} (\beta)$ & $0.0139668760758$ & $0.147740086$ & $9.88642$ & $106.322$ & $1071.4$ \\
         \hline
        $\delta_{100} ( \beta)$ & $0.0139688479485$ & $0.149783722$ & $17.3563$ & $329.338$ & $4983.8$ \\
        \hline
        Exact & $0.0139688479485$ & $0.149783722$ & $17.3563$ & $329.251$ & $5077.6$ \\
        \hline
	\end{tabular}

	\begin{tabular}{  c l l }
		Moments & $\beta = 0.1 $ & $\beta = 0.2 $    \\ 
		\hline
		50  & \textcolor{blue}{1.83 9}24$(10^{-4})$  &  \textcolor{blue}{7.0}2 337$(10^{-4})$  \\
  
		100 &  \textcolor{blue}{1.83 99}4$(10^{-4})$ & \textcolor{blue}{7.03 5}36$(10^{-4})$  \\
		    
        500 & \textcolor{blue}{1.83 994 677 2}27$(10^{-4})$  & \textcolor{blue}{7.03 568 26}1$(10^{-4})$  \\ 
        
        800 & \textcolor{blue}{1.83 994 677 220} 167$(10^{-4})$ & \textcolor{blue}{7.03 568 260 4}12$(10^{-4})$  \\ 
        
        1000 & \textcolor{blue}{1.83 994 677 220 3}47$(10^{-4})$ & \textcolor{blue}{7.03 568 260 4}74$(10^{-4})$    \\ 
        
        1500 & \textcolor{blue}{1.83 994 677 220 367} 084$(10^{-4})$ & \textcolor{blue}{7.03 568 260 484 }22$(10^{-4})$     \\
        
        2000 & \textcolor{blue}{1.83 994 677 220 367 06}4$(10^{-4})$ &  \textcolor{blue}{7.03 568 260 484 1}9$(10^{-4})$    \\
    
        \hline
        $P^{99}_{100} (\beta)$ & 1.83 994 677 220 361 577$(10^{-4})$ & 7.03 568 260 367 885$(10^{-4})$  \\
        \hline
         $\delta_{25} ( \beta)$ & 1.83 994 677 220 367 065$(10^{-4})$  & 7.03 568 260 484 163$(10^{-4})$ \\
        \hline
        Exact & 1.83 994 677 220 367 060$(10^{-4})$  & 7.03 568 260 484 187$(10^{-4})$ \\
        \hline
	\end{tabular}

	\begin{tabular}{  c l l l l l   }
		Moments & $ \beta = 10^{-2}$   \\ 
		\hline
	50 &  \textcolor{blue}{1.93 238 47}3$(10^{-6})$  &  &    \\
  
	100 & \textcolor{blue}{1.93 238 479 6}85$(10^{-6})$   \\

        500 & \textcolor{blue}{1.93 238 479 692 775 524 98}3 $(10^{-6})$ \\ 

        800 & \textcolor{blue}{1.93 238 479 692 775 524 980 51}9 646$(10^{-6})$  \\ 
        1000 & \textcolor{blue}{1.93 238 479 692 775 524 980 520 5}46$(10^{-6})$      \\ 
        
        1500 & \textcolor{blue}{1.93 238 479 692 775 524 980 520 558 841 7}2$(10^{-6})$      \\
        
        2000 & \textcolor{blue}{1.93 238 479 692 775 524 980 520 558 841 71}1$(10^{-6})$   \\
        
        \hline
        $P^{49}_{50} (\beta)$ & 1.93 238 479 692 775 524 980 520 558 841 700 571$(10^{-6})$   \\
        \hline
         $\delta_{35} ( \beta)$ & 1.93 238 479 692 775 524 980 520 558 841 710 583$(10^{-6})$ \\
        \hline
        Exact & 1.93 238 479 692 775 524 980 520 558 841 710 582$(10^{-6})$ \\
        \hline
	\end{tabular}

    \caption{Convergence of the extrapolant \eqref{hilg} constructed from the divergent expansion \eqref{gagah} for the integral $f_{0}(\beta)$ in the representation \eqref{oin}. $P^{N}_{M} (\beta)$ is a Pade approximant constructed from the divergent expansion \eqref{gagah} at 3000-digit precision using $N+M+1$ of the positive-power moments, $\mu^{(0)}_{2k}$, in equation \eqref{gigh}. We computed $\delta_{n} (\beta)$ from \cite{jen} using $n+1$ moments. The exact result is computed from the closed-form \eqref{impin}. }
	\label{hinglab}
\end{table}

\begin{table}
	\begin{tabular}{l l l l l l l}
		\hline
		Moments & $\beta = 10^{7}$ & $\beta = 10^{8}$ & $\beta = 10^{12}$ &  $\beta = 10^{15}$ & $\beta = 10^{17}$  \\ 
		\hline
  
		100 & \textcolor{blue}{1.}403$(10^{7})$  & 1.656$(10^{8})$ & 2.669$(10^{12})$ & 3.429$(10^{15})$ & 3.936$(10^{17})$  \\
  
		200 & \textcolor{blue}{1.}589$(10^{7})$ & 1.882$(10^{8})$  & \textcolor{blue}{3}.054$(10^{12})$ & 3.932$(10^{15})$   & {4}.518$(10^{17})$   & \\

 	500 & \textcolor{blue}{1}.732$(10^{7})$ & \textcolor{blue}{2.}057$(10^{8})$ & \textcolor{blue}{3}.360$(10^{12})$ & \textcolor{blue}{4.}337$(10^{15})$ & {4}.988$(10^{17})$ &   \\ 

		800 & \textcolor{blue}{1.}782$(10^{7})$ & \textcolor{blue}{2.}119$(10^{8})$ & \textcolor{blue}{3}.470$(10^{12})$ & \textcolor{blue}{4.}484$(10^{15})$ & \textcolor{blue}{5}.160$(10^{17})$  &  \\ 

		1000 & \textcolor{blue}{1.}801$(10^{7})$ & \textcolor{blue}{2.}142$(10^{8})$ & \textcolor{blue}{3}.514$(10^{12})$ & \textcolor{blue}{4.}542$(10^{15})$  & \textcolor{blue}{5}.227$(10^{17})$ &  \\ 

            1500 & \textcolor{blue}{1}.829$(10^{7})$ & \textcolor{blue}{2.}179$(10^{8})$ & \textcolor{blue}{3}.579$(10^{12})$ & \textcolor{blue}{4.}630$(10^{15})$ & \textcolor{blue}{5}.330$(10^{17})$ &  \\ 

            2000 &  \textcolor{blue}{1}.845$(10^{7})$ & \textcolor{blue}{2.}200$(10^{8})$ & \textcolor{blue}{3}.618$(10^{12})$ & \textcolor{blue}{4.}682$(10^{15})$ & \textcolor{blue}{5}.392$(10^{17})$ &  \\
		\hline
	    $P^{999}_{1000} (\beta)$ & $2.088(10^6)$  & 2.088$(10^{7})$ & 2.088$(10^{11})$ & 2.088$(10^{14})$  & 2.088$(10^{16})$\\
     
        $P^{49}_{50} (\beta)$ & $1.435(10^6)$  & $1.435(10^{7})$ & $1.435(10^{11})$ & $1.435(10^{14})$  & $1.435(10^{16})$\\
        \hline
        $\delta_{499} (\beta)$ & $1.469(10^{7})$& 1.512$(10^{8})$ & 4.187$(10^{14})$& 4.172$(10^{20})$  & 4.172$(10^{24})$  & \\
       
        $\delta_{100} (\beta)$ & $1.897(10^{7})$& $1.048(10^{9})$ & $9.539(10^{16})$& $9.539(10^{22})$  & $9.539(10^{26})$  & \\
        \hline
        Exact & $1.925(10^{7})$& 2.307$(10^{8})$ & 3.841$(10^{12})$& 4.992$(10^{15})$  & 5.760$(10^{17})$  \\
        \hline
	\end{tabular}

    	\begin{tabular}{l l l l l l l}
		
		Moments & $\beta = 1$ & $\beta = 4$ & $\beta = 10$ &  $\beta = 10^2$ & $\beta = 10^3$  \\ 
		\hline
  
		100 & \textcolor{blue}{1.6}394$(10^{-2})$  & \textcolor{blue}{0.1}7938 & 0.77725 & \textcolor{blue}{2}1.646 & \textcolor{blue}{4}18.1 \\
  
		200 & $\textcolor{blue}{1.645}239(10^{-2})$ & $\textcolor{blue}{0.182}02$  & $0.79805$ & $\textcolor{blue}{2}2.989$& $\textcolor{blue}{4}56.0$\\

 	500 & $\textcolor{blue}{1.6459}611(10^{-2})$ & $\textcolor{blue}{0.182}64$ & $\textcolor{blue}{0.80}451$ & $\textcolor{blue}{23}.644$ & $\textcolor{blue}{4}79.1$  \\ 

		800 & $\textcolor{blue}{1.64598}574(10^{-2})$ & $\textcolor{blue}{0.182}69$ & $\textcolor{blue}{0.805}28$ & $\textcolor{blue}{23}.782$ & $\textcolor{blue}{4}85.4$ \\ 

		1000 & $\textcolor{blue}{1.64598}813(10^{-2})$ & $\textcolor{blue}{0.182}696$  & $\textcolor{blue}{0.805}44$  & $\textcolor{blue}{23.}820$ & $\textcolor{blue}{4}87.5$ \\ 

            1500 & $\textcolor{blue}{1.645989}234(10^{-2})$ & $\textcolor{blue}{0.18270}17$  & $\textcolor{blue}{0.805}57$  & $\textcolor{blue}{23.}865$  & $\textcolor{blue}{49}0.2$ \\ 

            2000 & $\textcolor{blue}{1.6459893}58(10^{-2})$ & $\textcolor{blue}{0.18270}29$ & $\textcolor{blue}{0.8056}1$ &  $\textcolor{blue}{23}.882$ & $\textcolor{blue}{49}1.5$ \\
		\hline
	    $P^{999}_{1000} (\beta)$ & 1.645988828$(10^{-2})$  & 0.1825789  & 0.79637 & 17.589 & 204.9\\
        $P^{49}_{50} (\beta)$ & $1.645771086(10^{-2})$ & $0.1802697$ & $0.74632$ & $13.085$ & 142.1\\
        \hline
        $\delta_{30} (\beta)$ & $1.645989388(10^{-2})$ & $0.1827035$ & $0.80563$ & $23.961$ & $466.7$ & \\
   
        \hline
        Exact & $1.645989388(10^{-2})$ & $0.1827035$ & $0.80564$ & $23.907$ & 494.5 & \\

        \hline
	\end{tabular}

    \caption{Convergence of the extrapolant \eqref{hilg} constructed from the divergent expansion \eqref{gagah} for $f_{\frac{1}{2}}(\beta)$. $P^{N}_{M} (\beta)$ is a Pade approximant constructed from the divergent expansion \eqref{gagah} at 3000-digit precision. We obtained $\delta_{n}( \beta)$ from \cite[eq 4]{jen}.  The exact result is computed from the closed-form \eqref{kilat}.}
	\label{hngaa}
\end{table}

The expansion in equation \eqref{som} can be shown to be absolutely convergent (the proof is similar to that of Theorem 3.4 in \cite{galapon2}). This convergence and the existence of the limit \eqref{pisik} justifies  interchanging the limit operation with the infinite sum in equation \eqref{alad}. Meanwhile,
the second term $\Delta(\beta)$ is missed by merely performing a formal term-by-term integration followed by a regularization of the divergent integrals by Hadamard's finite part. It provides the dominant behavior for the generalized Stieltjes integral as $\beta\to\infty$. This term originates from contributions of the poles $z =\pm i/\sqrt{\beta}$ interior to the contour $\mathrm{C}$ in figure \ref{tear}  which results from the uniformity condition imposed in equation \eqref{gabt}.

By contrast, the finite-part integration \eqref{gidak} of integral representation for the Heisenberg-Euler Lagrangian involves term-by-term integration over a finite number of divergent integrals so that no uniformity condition is imposed and consequently, the contour $\mathrm{C}$ in figure \ref{tear2} excludes any of the poles of the integrand along the imaginary axis.

Hence, we arrive at a convergent expansion for $f_s(\beta)$ from the result \eqref{som} by finite-part integration,
\begin{equation}\label{hilg}
    f_s(\beta) = \sum_{k=0}^{\infty} 
\frac{(-1)^k}{\beta^{k-1}} \mu_{-(2k+2)}^{(s)} + \beta \Delta(\beta).
\end{equation}
In order for the expansion \eqref{hilg} to extrapolate the divergent weak-field expansion \eqref{gagah} to the non-perturbative regime, $\beta\to\infty$ along the real line, we incorporate the leading-order behavior \eqref{mirt} through the term $\Delta(\beta)$ given in equation \eqref{gibad}. To this end, we require the reconstruction of the function $\rho_s(x)$ from the positive-power moments $\mu_{2k}^{(s)}$ in equation \eqref{gigh} to be of the form $\rho_s(x) = x g_s(x)$ where $g_s(0)\neq 0$ and has an entire complex extension $g_s(z)$. This is done by expanding $g_s(x)$ as a generalized Fourier series expansion in terms of the Laguerre polynomials, $L_m(x)$ \cite{four},
\begin{equation}
    g_s(x) = \sum_{m=0}^{\infty} c_m^{(s)} \psi_m(x), \,\,\,\, \psi_m(x) = e^{-x/2}L_m(x).
\end{equation}
The basis functions $\psi_m(x)$ obey the orthonormality relation \cite{ortho},
\begin{equation}
    \int_{0}^{\infty} \left(e^{-x/2}L_m(x)\right)\,\left(e^{-x/2}L_n(x)\right) \mathrm{d}x = \delta_{n,m}
\end{equation}
and the Laguerre polynomials are given by \cite{laguerree}, 
\begin{equation}
    L_m(x) = \sum_{k=0}^{m} \frac{\left(-m\right)_k x^{k}}{(k!)^2} = m!\sum_{k=0}^{m}\frac{(-x)^k } {(k!)^2\,(m-k)!} .
\end{equation}
So that the reconstruction of $g_s(x)$ takes the form,
\begin{equation}\label{mity}
    g_s(x) = e^{-x/2}\sum_{m=0}^{\infty} c_m^{(s)} m! \sum_{k=0}^{m}\frac{(-x)^k}{(k!)^2 (m-k)!}.
\end{equation}
Hence with the form $\rho_s(x) = x g_s(x)$, the second term of the right-hand side of equation \eqref{hilg} can simulate the leading-order behavior \eqref{mirt}.

The first $d+1$ expansion coefficients, $c_m^{(s)}$, in the reconstruction \eqref{mity} are then computed by imposing the moment condition \eqref{gigh}. This results to a system of linear equations,
\begin{align}\label{sugr}
    a_{n+2}^{(s)} = \sum_{m=0}^{d} c_m^{(s)} P(n,m)
\end{align}
where the matrix $P(n,m)$ is given by
\begin{equation}
    P(n,m) = m! 2^{2n+2}\sum_{k=0}^{m}\frac{(-2)^{k}\,(2n+k+1)!}{(k!)^{2} (m-k)!}.
\end{equation}
We solve this system using the LU factorization method and solver provided by the C++ Eigen 3 library \cite{eigenweb}. We also used arbitrary precision data types from \cite{mpfr} and the C++ Boost Multiprecision libraries to represent the PT coefficients $a_{n+2}^{(s)}$ and perform our computations in arbitrary precision. 

We then substitute the reconstruction $\rho_s(x) = x g_s(x)$, where $g_s(x)$ is given by equation \eqref{mity}, to the expansion \eqref{hilg} so that the first term evaluates to
\begin{align}\label{gipoy}
\sum_{k=0}^{\infty}\frac{(-1)^{k}}{\beta^{k-1}} \mu_{-(2k+2)}^{(s)} = \sum_{k=0}^{\lfloor\frac{d-1}{2}\rfloor} \frac{(-1)^{k}}{\beta^{k-1}} \left(I_k + J_k + L_k \right)
     + \sum_{k=\lfloor\frac{d-1}{2}\rfloor+1}^{\infty} \frac{(-1)^{k}}{\beta^{k-1}} M_k,
\end{align}
where $\lfloor x \rfloor$ is the floor function and the coefficients are given by,
\begin{equation}
    I_k = \sum_{m=0}^{2k}c_m^{(s)} m! \sum_{l=0}^{m}\frac{(-1)^{l}}{(l!)^2 (m-l)!}\bbint{0}{\infty}\frac{e^{-x/2}}{x^{2k+1-l}}\mathrm{d}x,
\end{equation}
\begin{equation}
    J_k = \sum_{m=2k+1}^{d}c_m^{(s)} m! \sum_{l=0}^{2k}\frac{(-1)^{l}}{(l!)^2 (m-l)!}\bbint{0}{\infty}\frac{e^{-x/2}}{x^{2k+1-l}}\mathrm{d}x,
\end{equation}
\begin{equation}
    L_k = \sum_{m=2k+1}^{d} c_m^{(s)} m! \sum_{l=2k+1}^{m}\frac{(-1)^{l}\,(l-2k-1)!\,2^{l-2k}}{(l!)^2 (m-l)!},
\end{equation}
and
\begin{equation}
    M_k = \sum_{m=0}^{d}c_m^{(s)} m! \sum_{l=0}^{m}\frac{(-1)^{l}}{(l!)^2 (m-l)!}\bbint{0}{\infty}\frac{e^{-x/2}}{x^{2k+1-l}}\mathrm{d}x.
\end{equation}
The finite part integrals appearing in theses terms are given by equation \eqref{ighi},
\begin{equation}
    \bbint{0}{\infty}\frac{e^{-x/2}}{x^{2k+1-l}}\mathrm{d}x = \frac{(-1)^{1-l}\left(\frac{1}{2}\right)^{2k-l}}{(2k-l)!}\left(\ln\left(\frac{1}{2}\right)-\psi(2k+1-l)\right).
\end{equation}

The convergence of the expansion \eqref{hilg} across a wide range of field strengths is summarized in table \ref{hinglab} for the spin-0 case and in table \ref{hngaa} for spin-$\frac{1}{2}$. The result presented along each row is computed by adding up to $ k = 2d$ terms of the convergent expansion in equation \eqref{gipoy}, where $d+1$ is the number of moments used. As a rule, the working precision in digits at which we carry out the computation equals the number of moments $a_{n+2}^{(s)}$ used as inputs in the system of linear equations \eqref{sugr}. 

In the strong-field limit for both spin cases, our result reproduces the first few digits of the exact value. In the weak to intermediate field strength, the extrapolant \eqref{hilg} exhibits an excellent agreement with the exact values as well as with the Pad\'e approximant, $P^{N}_{M} (\beta)$, and the nonlinear sequence transformation, $\delta_{n}(\beta) $, from \cite[eq 4]{jen}. $P^{N}_{M} (\beta)$ is constructed from the PT expansion \eqref{gagah} at 3000-digit working precision using $N+M+1$ expansion coefficients while $\delta_{n}(\beta)$ uses $n+1$. Both these latter methods perform well in the weak to intermediate regimes but fail to extrapolate the divergent PT expansion \eqref{gagah} well into the $\beta\to\infty$ regime regardless of how many coefficients are used as inputs in their construction. This can be traced to the inability of these methods to incorporate the precise logarithmic strong-field behaviors \eqref{mirt} of the Heisenberg-Euler Lagrangian. The Pad\'e approximant for instance exhibits an integer power leading-order behavior, $P^{N}_{M}(\beta)\sim\beta^{N-M}$ as $\beta\to\infty$.

We also compared the terms in the expansion \eqref{hilg} in the case of spin-0 particles in table \ref{boki}. Both terms are comparable in magnitude and are therefore computationally relevant across all values of the parameter $\beta$ even well into the non-perturbative regime, $\beta\to \infty$, where the second term begins to dominate the first albeit rather slowly due to the logarithmic growth.
Another important feature to note is that the second term contains the sampling of the reconstruction \eqref{mity} of the function $\rho_s(x) = x g_s(x)$. This implies that convergence of the extrapolant \eqref{hilg} in the strong-field regime can be improved or undermined by the point-wise convergence of the reconstruction \eqref{mity}. In the context of the reconstruction scheme chosen here as a generalized Fourier series \eqref{mity}, Gibbs phenomenon, which is the spurious oscillation that arises when a zero or a singularity is not simulated by the reconstruction, is of particular relevance. Here, we are prompted to introduce a zero at $x=0$ in the reconstruction of $\rho_s(x)$ on the basis of the leading-order behavior \eqref{mirt} that we wish to incorporate in the expansion for $f_s(\beta)$.

\subsection{Self-Dual Background}
The integral appearing in the exact formulation \eqref{muska} possesses the following weak-field perturbation expansion \cite{self_dual2}
\begin{align}\label{upgit}
   f_{\mathrm{SD}}(\beta)  
    \sim  \beta \sum_{k=0}^{\infty}-\frac{(-1)^k B_{2k+4}}{(2k+2) (2k+4)} (-\beta)^{k}, \qquad \beta\to 0.
\end{align}
The expansion coefficients alternate in sign and possess a leading-order growth $(2k)!$ as $k\to\infty$. We tabulated the partial sums of this divergent PT expansion in table \ref{biak} for small values of the parameter $\beta$. The asymptotic nature of this expansion manifests with the initial rapid convergence of the sequence of partial sums for small value of $\beta$ before they eventually diverge. Furthermore, the integral $f_{\mathrm{SD}}(\beta)$, possesses the leading-order behavior in the non-perturbative regime \cite{dunne, self_dual1, self_dual2},
\begin{equation}\label{kindat}
    f_{\mathrm{SD}}(\beta)\sim \ln \beta \qquad \beta\to\infty.
\end{equation}

\begin{table}
	\begin{tabular}{ l l l l l l }
		\hline
		$d$   &  $\beta=10^{-2}$ & $\beta=0.1$&  $\beta=0.2$\\ 
		\hline
		1   & $\textcolor{blue}{4.156}746(10^{-5})$ & $\textcolor{blue}{4.0}675(10^{-4})$ & $\textcolor{blue}{7.9}365(10^{-4})$ \\
            2   & $\textcolor{blue}{4.1568}15476(10^{-5})$  & $\textcolor{blue}{4.07}440(10^{-4})$ & $\textcolor{blue}{7.9}921(10^{-4})$\\
		5   & $\textcolor{blue}{4.1568145496}191(10^{-5})$ & $\textcolor{blue}{4.073}5993(10^{-4})$ &$\textcolor{blue}{7.9}792({10^{-4}})$\\
		9   & $\textcolor{blue}{4.156814549649}016(10^{-5})$ & $\textcolor{blue}{4.07361}247(10^{-4})$ 
                &$\textcolor{blue}{7.9}7150(10^{-4})$ \\
		20  & $\textcolor{blue}{4.1568145496490179111}20(10^{-5})$ & $\textcolor{blue}{4}.2517(10^{-4})$ & $4.111(10)$\\
		50  & \textcolor{blue}{4.15681454964901791}2493$(10^{-5})$ &  &\\
          \hline
		Exact & 4.1568145496490179111196$(10^{-5})$ & 4.0736197107$(10^{-4})$ & 7.981190$(10^{-4})$  \\
		\hline
  
	\end{tabular}
 
	\caption{Convergence of the partial sums of the perturbative expansion \eqref{upgit} for the integral $f_{\mathrm{SD}}(\beta)$ Heisenberg-Euler Lagrangian in the self-dual background. The exact result is computed from the closed-form \eqref{nopita}.}
 \label{biak}
\end{table}

\begin{table}
	\begin{tabular}{lllll}
		\hline
		$\beta$  & first term &  second term & $f_0(\beta)$ \\ 
		\hline
		$0.1$ & $2.430021989812885$ & $-2.42983799513566$ & $\underline{1.8399467722036706}4(10^{-4})$ \\
		$1.0$ & $-3.648351466(10^{-2}) $ & $5.045236261518(10^{-2})$ & $\underline{0.0139688479}511$  \\
		$4.0$ & $-1.1382129130$ & $1.287996639$ & $\underline{0.14978372}6$  \\
		$10^{2}$ & $-26.757661$ & $44.114088$ & $\underline{17.356}4$ \\
		$10^{4}$ & $-2665.8699$ & $7743.7243$ & $\underline{5077}.9$  \\
		$10^{7}$ & $-2.6657736(10^6)$ & $1.3436300(10^{7})$ & $\underline{1.07}71(10^{7})$ \\
  	$10^{18}$ & $-2.6657697^(10^{17})$ & $3.4544849^(10^{18})$ & $\underline{3.187}91(10^{18})$ \\ 
		$10^{21}$ & $-2.6657697(10^{20})$ & $4.0302324(10^{21})$& $\underline{3.763}66(10^{21})$ \\ 
		\hline
	\end{tabular}
	\caption{Comparison of terms in the expansion \eqref{hilg} for $f_0(\beta)$ in the Heisenberg-Euler Lagrangian in the case of spin-0 particles at 2000-moment reconstruction. Both terms remain relevant across all parametric regimes even well into the strong field regime, $\beta \to\infty$. The digits underlined coincide with those of the exact result.}
	\label{boki}
\end{table}

 On the basis of this information, we map the PT coefficients in \eqref{upgit} to the positive-power moments, $\mu_{2k}$, of some positive function $\rho(x)$,
\begin{equation}
    -\frac{(-1)^{k} B_{2k+4}}{(2k+2) (2k+4)} = \mu_{2k} =  \int_{0}^{\infty}x^{2k}\rho(x)\mathrm{d}x,\,\,\,k=0,1,\dots
\end{equation}
So that we can sum the PT expansion \eqref{upgit} to a generalized Stieltjes integral \eqref{som}
\begin{align}
    f_{\mathrm{SD}}(\beta) = \beta\int_{0}^{\infty}\frac{\rho(x)}{1+\beta x^2}\mathrm{d}x.
\end{align}
Then by finite-part integration, we obtain an expansion for $f_{\mathrm{SD}}(\beta)$ from \eqref{som},
\begin{align}\label{niop}
    f_{\mathrm{SD}}(\beta) = \sum_{k=0}^{\infty} 
\frac{(-1)^k}{\beta^{k}} \mu_{-(2k+2)} +  \Delta(\beta).
\end{align}
As in the previous cases, we accommodate the leading-order behavior \eqref{kindat} through the second term $\Delta(\beta)$ by requiring the reconstruction of the positive function $\rho(x)$  in the form $\rho(x) = xg(x)$ where $g(x)$ is again given by \eqref{mity} and the reconstruction scheme proceeds as in the previous cases. The first term in the right-hand side of the expansion \eqref{niop}, is computed in exactly the same way as in equation \eqref{gipoy} in the previous examples. 

The convergence of the extrapolant \eqref{niop} of the divergent expansion \eqref{upgit} is summarized in table \ref{hoyaa} for various values of the parameter $\beta$ in the strong-field regime where the expansion reproduces the first few digits of the exact value computed from the closed-form \eqref{nopita}. As in the previous cases, the Pad\'e approximant $P^{N}_{M} (\beta)$ and the nonlinear sequence transformation $\delta_{n} (\beta)$ becomes unreliable as $\beta\to\infty$. 

\begin{table}
	\begin{tabular}{l l l l l l l l l}
		\hline
		Moments &  $\beta = 10^{7}$ & $\beta = 10^{9}$ & $\beta = 10^{13}$ & $\beta = 10^{18}$ &  $\beta = 10^{19}$ & $\beta = 10^{20}$  \\ 
		\hline
  
		100 & 0.40507 &  0.55496 & 0.85492 & \textcolor{blue}{1.}22989 & \textcolor{blue}{1.}30488 & \textcolor{blue}{1.}37987  \\
  
		200 & \textcolor{blue}{0.5}1549 & 0.71213 & \textcolor{blue}{1.}10570 & \textcolor{blue}{1.5}9769 & \textcolor{blue}{1.6}9609   & \textcolor{blue}{1.7}9449 & \\

		800 & \textcolor{blue}{0.50}269 & \textcolor{blue}{0.69}251  &\textcolor{blue}{1.0}7240 & \textcolor{blue}{1.5}4729 & \textcolor{blue}{1.6}4227 & \textcolor{blue}{1.7}3724  &  \\ 

		1000  & \textcolor{blue}{0.50}250 & \textcolor{blue}{0.69}223 &\textcolor{blue}{1.0}7194 & \textcolor{blue}{1.5}4660 & \textcolor{blue}{1.6}4153  & \textcolor{blue}{1.7}3647 &  \\ 

            1500  & \textcolor{blue}{0.50}598 & \textcolor{blue}{0.69}755 &\textcolor{blue}{1.08}095 & \textcolor{blue}{1.56}022& \textcolor{blue}{1.65}608 & \textcolor{blue}{1.75}193 &  \\ 

            2000  & \textcolor{blue}{0.506}86 & \textcolor{blue}{0.698}90 &\textcolor{blue}{1.08}326 & \textcolor{blue}{1.56}374 & \textcolor{blue}{1.65}983 & \textcolor{blue}{1.75}593 &  \\
		\hline
	    $P^{999}_{1000} (\beta)$ & 0.08926 & 0.08926 & 0.08926 & 0.08926 & 0.08926  & 0.08926 \\
        
        $P^{49}_{50} (\beta)$ & 0.06347 & 0.06347 & 0.06347 & 0.06347 & 0.06347  & 0.06347 \\
        \hline
         $\delta_{999} (\beta)$ & 0.46570 & 0.31717 &-2.5$(10^2)$ & -2.5$(10^7)$ & -2.5$(10^8)$ & -2.5$(10^9)$ & \\
           
         $\delta_{100} (\beta)$ & 0.34984 & 5.44784 &$5.15(10^4)$ & $5.15(10^9)$& $5.15(10^{10})$& $5.15(10^{11})$ & \\
        \hline
        Exact  & 0.50632 & 0.69806 & 1.08181 & 1.56152 & 1.65746 & 1.75340 & \\
        \hline

	\end{tabular}

    \caption{Convergence of the extrapolant \eqref{niop} constructed from the divergent expansion \eqref{upgit} for the integral $f_{\mathrm{SD}}(\beta)$ in the regime $\beta\to\infty$. $P^{N}_{M} (\beta)$ is a Pad\'e approximant computed from the PT series \eqref{upgit} at 3000-digit working precision. $\delta_{n} (\beta)$ is obtained using the result in \cite[eq 4]{jen}. The exact value is computed from the closed-form \eqref{nopita}.}
	\label{hoyaa}
\end{table}

\section{Conclusion}\label{conclusion}
In this paper, we proposed a prescription based on the method of finite-part integration of the generalized Stieltjes integral to sum divergent PT series expansion with coefficients that we map to the positive-power moments, $\mu_{2k} =\int_{0}^{\infty}x^{2k}\rho(x)\mathrm{d}x$, of some positive function $\rho(x)$. We applied the summation procedure on the divergent weak-field expansions of exact integral representations for various Heisenberg-Euler Lagrangians from QED in the case of a constant magnetic and magnetic-like self-dual background. In each of these examples, the procedure allowed us to transform the divergent alternating weak-field expansion into a novel convergent expansion in  inverse powers of perturbation parameter plus a correction term that led us to incorporate the known logarithmic leading-order behavior in the strong-field regime. This enabled us to construct extrapolants which can be used across a wide range of values for the field strength with considerable accuracy. Furthermore, we also showed how the method of finite-part integration can be used to evaluate in closed-form the exact integral representations of the Heisenberg-Euler Lagrangians. 

There are a few ways to improve the summation prescription we devised in this paper. We pointed out how the point-wise convergence of the solution to the underlying Stieltjes moment problem could potentially undermine or improve the convergence of the expansion. In this regard, future work on the subject could investigate procedures for solving the Stieltjes moment problem that could offer better point-wise convergence with less number of positive-power moments $\mu_k$ used as inputs. This could enhance the potential of the prescription to take on problems in certain applications where the perturbation coefficients are scarce. A promising alternative proposed in \cite{mead1984maximum} is an information-theoretic approach and is based on the maximization of the Shannon entropy functional. 

 In the case of a purely electric background so that $\beta = - \kappa$, $\kappa = (e E)^2/m^4 >0$, and the expansion \eqref{gagah} becomes non-alternating, the prescription in \eqref{gigh} sums the expansion \eqref{gagah} formally to
 \begin{align}
     f_s(\kappa) = \kappa \int_{0}^{\infty} \frac{\rho_s(x)}{1/\kappa-x^{2}} \mathrm{d}x = \kappa H_s(\kappa).
 \end{align}
 The formal integral $H_{\kappa}(\kappa)$ is then evaluated using a suitable prescription as, 
\begin{align}\label{ugay}
        H_s(\kappa)  = \mathrm{PV}\int_{0}^{\infty}\frac{\,\rho_s(x)}{1/\kappa- x^{2}}\mathrm{d}x \pm \frac{\pi i \sqrt{\kappa}}{2} \rho_s\left(\frac{1}{\sqrt{\kappa}}\right),
\end{align}
where the sign of the second term is a non-perturbative ambiguity and corresponds to how the the pole at $z = 1/\sqrt{\kappa}$ is evaded. In the case of the Heisenberg-Euler Lagrangian, the sign must be chosen so that the final result yields the positive value for the imaginary part since this gives the particle-antiparticle pair production rate. The first term in the right-hand of equation \eqref{ugay} is a Cauchy principal value integral and is evaluated explicitly using finite-part integration so that, 
\begin{align}\label{wikop}
    H_s(\kappa) = -\sum_{n=0}^{\infty} 
\frac{\mu^{(s)}_{-(2n+2)}}{\kappa^n} -\Delta(\kappa) \pm \frac{\pi i \sqrt{\kappa}}{2} \rho_s\left(\frac{1}{\sqrt{\kappa}}\right).
\end{align}
The negative-power moments, $\mu^{(s)}_{-(2n+2)}$, are the finite part integrals \eqref{pigil} and the term $\Delta(\kappa)$ missed by a naive term-by-term integration is given by,
\begin{align}
    \Delta(\kappa) =  \frac{\sqrt{\kappa}}{2}\ln\left(\sqrt{\kappa}\right)\left(\rho_s\left(\frac{1}{\sqrt{\kappa}}\right)-\rho_s\left(-\frac{1}{\sqrt{\kappa}}\right)\right).
\end{align}
This term provides the known leading-order behavior in the the strong-field regime for the real part of the Heisenberg-Euler Lagrangian. The regularization \eqref{wikop} can also be derived by performing analytic continuation, $S_s(\beta\to -\kappa) =H_s(\kappa)$, in the expansion \eqref{som} to the branch cut of $S_s(\kappa)$ along the negative real axis of the complex-$\beta$ plane. In which case, the sign ambiguity of the second term of the right-hand side of equation \eqref{wikop} corresponds to whether the branch cut is approached from above ($+$) or below ($-$). The details of the computation and the results are given in \cite{nonalternating}. 

\section*{Acknowledgments}
We acknowledge the Computing and Archiving Research Environment (COARE) of the Department of Science and Technology's Advanced Science and Technology Institute (DOST-ASTI) for providing access to their High-Performance Computing (HPC) facility. This work is funded by the University of the Philippines System through the Enhanced Creative Work Research Grant (ECWRG 2019-05-R). C.D. Tica acknowledges the Department of Science and Technology-Science Education Institute (DOST-SEI) for the scholarship grant under DOST ASTHRDP-NSC.

\end{document}